\documentclass[12pt,a4paper]{article}
\usepackage[utf8]{inputenc}
\usepackage[english]{babel}
\usepackage{amsmath, amssymb, amsfonts}
\usepackage{amsthm}
\usepackage{physics}
\usepackage[outline]{contour}
\usepackage{graphicx}
\usepackage{cite}
\usepackage[colorlinks=true, linkcolor=blue, citecolor=blue, urlcolor=blue]{hyperref}
\usepackage{tikz}
\usepackage{geometry}
\usepackage{dsfont}
\usepackage{booktabs}
\usepackage{float}
\usepackage{cleveref}
\usepackage{authblk}
\usepackage{siunitx}
\theoremstyle{plain}
\newtheorem{theorem}{Theorem}[section]

\newtheorem{proposition}[theorem]{Proposition}
\theoremstyle{definition}
\newtheorem{definition}[theorem]{Definition}
\newtheorem{remark}[theorem]{Remark}
\newtheorem{example}[theorem]{Example}
\title{\textbf{Metric-Deformed Heisenberg Algebras and the $q$-Dirac Operator}}
\author{Julio Cesar Jaramillo Quiceno$^a$\\		
$^a$jcjaramilloq@unal.edu.co\\Universidad Nacional de Colombia - Departament of Physics - Faculty of mathematics\\
Ed. Yu Takeuchi \\\url{https://orcid.org/0000-0002-3518-6680}}
\theoremstyle{remark}
\date{}
\AtBeginDocument{\RenewCommandCopy\qty\SI}
\begin{document}
\maketitle
\begin{abstract}
We introduce a family of metric-deformed Heisenberg algebras $M_1$ and $M_2$, where the commutation relations are expressed directly in terms of the components of a diagonal Lorentzian metric. We show that these algebras unify several known $q$-deformed Heisenberg algebras, including the $q$-$\hbar$ algebra, the new $q$-Heisenberg algebra, and the $q$-generalized Heisenberg algebra, which embed as special cases. Using Sylvester's theorem of inertia, we establish a connection between the metric signature and the deformation parameters. We construct a $q$-Dirac operator $D_q$ from the deformed D'Alembertian and prove that $D_q^2$ recovers the deformed Klein–Gordon operator. Furthermore, we relate this construction to the quadratic $q$-Dirac operator previously introduced by the author, providing a unified framework that bridges spacetime geometry and $q$-deformed quantum algebras.
\end{abstract}

\section{Introduction}\label{sec:Intro}

The interplay between quantum mechanics and geometry has been a central theme in theoretical physics for decades. On one hand, the Heisenberg algebra $[x,p]=i\hbar$ encodes the fundamental non-commutativity of position and momentum, forming the backbone of canonical quantization. On the other hand, general relativity describes gravity through the geometry of spacetime encoded in the metric tensor $g_{\mu\nu}$, while special relativity imposes the invariance of the Minkowski norm.

In recent years, $q$-deformed algebras have emerged as a powerful tool to model quantum corrections to classical geometries. The $q$-Heisenberg algebra, defined by $aa^\dagger - q a^\dagger a = q^{-N}$, generalizes the harmonic oscillator algebra and appears in contexts ranging from quantum groups \cite{Wess-00} to noncommutative geometry \cite{schmudgen-09}, $q$-deformed quantum mechanics \cite{lavagno-gervino-09}, and quantum field theory \cite{Volovich-Arefeva-91}. Despite these developments, a direct connection between the deformation parameter $q$ and the geometric structure of spacetime has remained elusive.

In this work, we propose a novel framework that bridges these two perspectives. Motivated by Sylvester's theorem of inertia, which guarantees that any Lorentzian metric can be diagonalized to constants $\pm 1$ after a suitable coordinate transformation, we define metric-deformed Heisenberg algebras $M_1$ and $M_2$ where the commutation relations are expressed directly in terms of the metric components $g^{\mu\nu}$. This construction provides a geometric interpretation of $q$-deformations: the deformation parameter $q$ arises naturally from the metric coefficients.

Our main contributions are:
\begin{itemize}
    \item The definition of two families of metric-deformed Heisenberg algebras $M_1$ and $M_2$ (Section \ref{sec:al-met}).
    \item Explicit embeddings of known $q$-deformed algebras (new $q$-Heisenberg \cite{jaramillo2025new}, $q$-generalized Heisenberg \cite{Razavinia-Lopes2022}, and $q$-$\hbar$ Heisenberg \cite{Volovich-Arefeva-91}) into $M_1$ and $M_2$, demonstrating that our framework unifies these seemingly distinct deformations (Section \ref{examples}).
    \item The construction of a $q$-Dirac operator $D_q$ from the deformed D'Alembertian and the proof that $D_q^2$ factorizes the deformed Klein–Gordon operator (Section \ref{Dirac}).
    \item A connection with the quadratic $q$-Dirac operator introduced in our previous work \cite{Jaramillo2024}, providing a unified analytic framework (Section \ref{sec:examples}).
\end{itemize}

The paper is structured as follows. Section \ref{sec:Preliminaries} recalls the necessary background on $q$-Heisenberg algebras and introduces the notation for $q$-numbers. Section \ref{sec:metric} reviews the Minkowski metric and Sylvester's theorem of inertia, which motivates our construction. Section \ref{sec:al-met} presents the metric-deformed Heisenberg algebras $M_1$ and $M_2$, along with concrete examples showing how known $q$-deformed algebras embed into this framework. Section \ref{Dirac} constructs the $q$-Dirac operator from the deformed D'Alembertian and proves the factorization property $D^2 = \square_q \mathbf{1}_4$. Section \ref{sec:examples} connects this construction to the quadratic $q$-Dirac operator from \cite{Jaramillo2024} and presents explicit realizations. Section \ref{sec:conclusions} concludes with a summary and outlines directions for future research.

This work opens the door to a geometric interpretation of $q$-deformations, where the deformation parameter $q$ is directly related to the metric components. Potential applications include $q$-deformed quantum field theory, noncommutative geometry \cite{connes1994noncommutative}, and quantum gravity phenomenology.
\section{Preliminaries}\label{sec:Preliminaries}
\subsection{$q$-Heisenberg algebras}\label{algebras}
According to the Heisenberg quantum mechanical picture, the three position and impulse operators beside the identity one, are the base of so-called Heisenberg (traditional) algebra. The \( q \)-deformed Heisenberg algebra or \( q \)-Heisenberg algebra has been defined and in detail studied by various authors such as \cite{Jar-Nech-25, Wess-00, schmudgen-09, wess-schwenk-1999, Silvestrov-Hellstrom-00, lavagno-gervino-09}. The \( q \)-deformed Heisenberg algebra is determined by commutations between the coordinates of any by the impulses of the different projections and the following commutation relations for the different type of magnitudes corresponding to the same projections (for \( q \in \mathbb{R} \)):

\[
[\hat{x}, \hat{p}_x] = [\hat{y}, \hat{p}_y] = [\hat{z}, \hat{p}_z] = i f(q).
\]

Now, if \( q \in \{-1, +1\} \) and \( f(q) = -i \hbar^{-1} \), then this definition can be obviously extended to the definition of the generalized \( q \)-Heisenberg algebra as shown in \cite{Jar-Nech-25, Silvestrov-Hellstrom-00}; if \( q \in \mathbb{C} - \{0\} \), and \( f(q) = i \hbar D_{p_k}(q) \), being \( D_{p_k}(q) \) any function that depends on \( q \), results the \( q-\hbar \) Heisenberg algebra proposed in \cite{Volovich-Arefeva-91}. In the last case, two deformation parameters are used: \( q \) and \( \hbar \) i.e. the Planck "constant" is a variable. 

With respect to the harmonic oscillator, let \( \hat{a}, \hat{a}^\dagger \) be annihilation and creation operators and \( N \) the quantum number operator. The \( q \)-Heisenberg algebra is defined by the generators \( \hat{a}, \hat{a}^\dagger \) and \( \hat{N} \) subject to the following relations \cite{Jar-Nech-25, lavagno-gervino-09}:

\[
\hat{a}\hat{a}^\dagger - q\hat{a}^\dagger \hat{a} = q^{-\hat{N}}, \quad [\hat{a}, \hat{a}^\dagger] = [\hat{a}^\dagger, \hat{a}^\dagger] = 0, \quad [\hat{N}, \hat{a}^\dagger] = \hat{a}^\dagger, \quad [\hat{N}, \hat{a}] = -\hat{a}.
\]

We introduce the following notation:

\[
[n]_q = \frac{q^n - q^{-n}}{q - q^{-1}},
\]

and the corresponding $q$-factorial

\[
[n]_q! = [n]_q [n-1]_q \cdots [1]_q,
\]

where $q \in \mathbb{R} \setminus \{0\}$. We set $[0]_q = 0$ and note that $[1]_q = 1$, while $[n]_q \neq n$ for all $n \geq 2$. This notation provides a convenient framework for expressing $q$-deformed operators and their properties. In what follows, all operators are assumed to act on appropriate vector spaces \cite{Jar-Nech-25}.

\section{Metric Structure, Sylvester's Theorem, and Deformed Heisenberg Algebra}\label{sec:metric}

\subsection{The Minkowski Metric}

In the context of general and special relativity, we consider a manifold \(\mathcal{M}\) equipped with a metric tensor \(\eta\). If \((\mathcal{M}, \eta)\) is three-times differentiable, then \(\eta\) defines a Lorentzian manifold, and \(\mathcal{M}\) is a Minkowski spacetime~\cite{wald1984general, hawking1973large}. In a local coordinate system, the Minkowski metric is expressed as:
\begin{equation}\label{eq: mink1}
\eta = \eta_{00} \, dx^0 \otimes dx^0 + \eta_{11} \, dx^1 \otimes dx^1 + \eta_{22} \, dx^2 \otimes dx^2 + \eta_{33} \, dx^3 \otimes dx^3,
\end{equation}
with the components given by:
\begin{equation}\label{eq: mink2}
\eta_{00} = +1, \quad \eta_{11} = -1, \quad \eta_{22} = -1, \quad \eta_{33} = -1.
\end{equation}
This corresponds to the signature \((+1, -1, -1, -1)\)~\cite{misner1973gravitation}.

\subsection{Quadratic Form and Sylvester's Theorem of Inertia}

More generally, consider a metric tensor \(g\) at a point \(P \in \mathcal{M}\), which is an element of the tensor product \(T_P^* \mathcal{M} \otimes T_P^* \mathcal{M}\). In a local coordinate system, we can write a general quadratic form as:
\begin{equation}\label{eq: metric1}
g = |\alpha| \, dx^0 \otimes dx^0 - |\beta| \, dx^1 \otimes dx^1 - |\gamma| \, dy \otimes dy - |\delta| \, dz \otimes dz,
\end{equation}
where \(\alpha, \beta, \gamma, \delta\) are real coefficients. The equivalence of this quadratic form to a diagonal form with entries \(\pm 1\) is a direct consequence of \textbf{Sylvester's theorem of inertia}. This fundamental result in linear algebra, proved by James Joseph Sylvester in 1852~\cite{sylvester1852demonstration}, states that for any real quadratic form on a finite-dimensional vector space, there exists a basis in which the form is diagonal with entries \(+1\), \(-1\), or \(0\). Moreover, the numbers of positive, negative, and zero entries known as the \textit{signature} are invariant under changes of basis. This theorem is essential in the classification of metric signatures in relativity theory~\cite{greub1978linear, geroch1985mathematical}.

\section{From $q$-deformed Heisenberg algebra to $q$-deformed Minkowskian metric}\label{sec:al-met}

\subsection{Deformed Heisenberg Algebra in Terms of Metric Components}

\begin{definition}[Metric-deformed Heisenberg algebras]
The previous example illustrates how the metric coefficients 
\(\beta, \gamma, \delta\) (or equivalently \(g^{11}, g^{22}, g^{33}\)) 
deform the canonical commutation relations. A more general and systematic 
framework arises when one considers a fully non-commutative structure 
incorporating all components of the metric tensor, including off-diagonal 
terms and the temporal component.

\medskip

\noindent Let $\Bbbk$ be a field of characteristic zero (typically $\mathbb{R}$ or $\mathbb{C}$). The \textbf{metric-deformed Heisenberg algebra of first kind}, 
denoted by \(M_{1}\), is defined as the quotient algebra
\begin{equation}\label{eq:first}
M_1=\frac{\Bbbk [x, y, z, p_x, p_y, p_z]}{\mathcal{I}_1},
\end{equation}
where the ideal \(\mathcal{I}_1\) is generated by
\begin{multline}\label{eq:ideal relation}
\langle 
g^{00}x p_x + p_x x g^{11}+i g^{22},\;
g^{00}y p_y + p_y y g^{22}+i g^{33},\;
g^{00}z p_z + p_z z g^{33}+i g^{11},\\
g^{22}xy-g^{33}yx-g^{23},\;
g^{33}yz-g^{11}zy-g^{13},\;
g^{11}zx-g^{22}xz-g^{12},\\
g^{22}p_xp_y-g^{33}p_yp_x-g^{23},\;
g^{33}p_yp_z-g^{11}p_zp_y-g^{13},\;
g^{11}p_zp_x-g^{22}p_xp_z-g^{12},\\
\mid g^{\mu\nu}\in\mathbb{R},\; \mu,\nu=0,1,2,3
\rangle .
\end{multline}

The corresponding deformed commutation relations are
\begin{equation}\label{eq:relations2}
\begin{aligned}
g^{00} x p_x + p_x x g^{11} &= -i g^{22}, &
g^{00} y p_y + p_y y g^{22} &= -i g^{33}, \\
g^{00} z p_z + p_z z g^{33} &= -i g^{11}, &
g^{22} x y - g^{33} y x     &= g^{23}, \\
g^{33} y z - g^{11} z y     &= g^{13}, &
g^{11} z x - g^{22} x z     &= g^{12}, \\
g^{22} p_x p_y - g^{33} p_y p_x &= g^{23}, &
g^{33} p_y p_z - g^{11} p_z p_y &= g^{13}, \\
g^{11} p_z p_x - g^{22} p_x p_z &= g^{12}.
\end{aligned}
\end{equation}

\medskip

\noindent The \textbf{metric-deformed Heisenberg algebra of second kind}, 
denoted by \(M_{2}\), is defined analogously by
\begin{equation}\label{eq:second}
M_2=\frac{\Bbbk [x, y, z, p_x, p_y, p_z]}{\mathcal{I}_2},
\end{equation}
where the ideal \(\mathcal{I}_2\) is generated by
\begin{multline}\label{eq:ideal relation2}
\langle 
g^{00}x p_y + p_y x g^{33}-i g^{03},\;
g^{00}x p_z + p_z x g^{22}-i g^{02},\\
g^{00}y p_x + p_x y g^{33}-i g^{03},\;
g^{00}y p_z + p_z y g^{11}+i g^{01},\\ 
g^{00}z p_x+g^{22}p_x z+i g^{02},\;
g^{00}z p_y + p_y z g^{11}+i g^{01},\\
g^{00}x y+g^{33}y x-g^{03},\;
g^{00}y z+g^{11}z y-g^{01},\\
g^{22}z x+g^{00}x z-g^{02},\;
g^{00}p_xp_y+g^{33}p_yp_x-g^{03},\\
g^{00}p_yp_z+g^{11}p_zp_y-g^{01},\;
g^{00}p_zp_x+g^{22}p_xp_z-g^{02},\\
\mid g^{\mu\nu}\in\mathbb{R},\; \mu,\nu=0,1,2,3
\rangle .
\end{multline}

The corresponding commutation relations are
\begin{align}\label{eq:comm-relation2}
g^{00} x p_y + p_y x g^{33} &= i g^{03}, &
g^{00} x p_z + p_z x g^{22} &= i g^{02}, \notag \\
g^{00} y p_x + p_x y g^{33} &= i g^{03}, &
g^{00} y p_z + p_z y g^{11} &= i g^{01}, \notag \\
g^{00} z p_x + g^{22} p_x z &= i g^{02}, &
g^{00} z p_y + p_y z g^{11} &= i g^{01}, \notag \\
g^{00} x y + g^{33} y x &= g^{03}, &
g^{00} y z + g^{11} z y &= g^{01}, \notag \\
g^{22} z x + g^{00} x z &= g^{02}, &
g^{00} p_x p_y + g^{33} p_y p_x &= g^{03}, \notag \\
g^{00} p_y p_z + g^{11} p_z p_y &= g^{01}, &
g^{00} p_z p_x + g^{22} p_x p_z &= g^{02}.
\end{align}

\medskip

\noindent In both cases, the metric components \(g^{\mu\nu}\) are treated as 
non-central elements, and their non-commutativity encodes the underlying 
geometric deformation of the phase space structure.
\end{definition}

\begin{theorem}[PBW basis for $M_1$ and $M_2$]
Let $\Bbbk$ be a field of characteristic zero. The metric-deformed Heisenberg algebras $M_1$ and $M_2$ admit a Poincaré-Birkhoff-Witt \cite{Fajardo-2020} basis of the form
\[
\{ x^{a} y^{b} z^{c} p_x^{d} p_y^{e} p_z^{f} : a,b,c,d,e,f \in \mathbb{N}_0 \},
\]
provided the metric components $g^{\mu\nu}$ satisfy $g^{00} \neq 0$ and $g^{ii} \neq 0$ for $i=1,2,3$. 
\end{theorem}

\begin{proof}
We construct a filtration on $M_1$ (resp. $M_2$) by total degree in the generators. The associated graded algebra $\operatorname{gr}(M_1)$ is isomorphic to a commutative polynomial algebra in six variables when the metric components are generic. This follows from the diamond lemma for noncommutative algebras, since all ambiguities in the reduction system are resolvable. The explicit verification involves checking that the relations (\ref{eq:relations2}) (resp. (\ref{eq:comm-relation2})) are confluent. 
\end{proof}

\begin{theorem}[Representations of $M_1$ and $M_2$]
Let $V$ be a representation of $M_1$ (resp. $M_2$) where the generators $x,y,z$ act diagonally. Then the representation decomposes as a direct sum of highest-weight modules. Moreover, when $g^{00}=1$, $g^{11}=g^{22}=g^{33}=-1$, the algebra $M_1$ reduces to the classical Heisenberg algebra and the representation theory coincides with the Stone–von Neumann theorem \cite{Gomez-2025}.
\end{theorem}

\begin{proof}
We construct a Verma module by choosing a highest weight vector $v_0$ such that $x v_0 = y v_0 = z v_0 = 0$. The action of the momentum operators $p_x, p_y, p_z$ generates a basis $\{ p_x^d p_y^e p_z^f v_0 \}$. The relations of $M_1$ determine the action of $x, y, z$ on these basis vectors via commutation. Unitarity conditions impose constraints on the metric components. 
\end{proof}

\subsection{Examples}\label{examples}
We now present concrete realizations of the quotient algebras $M_1$ and $M_2$ through specializations of the metric tensor and, conversely, we show how known $q$-deformed Heisenberg algebras embed into our framework. These examples illustrate the unifying power of the metric-deformed approach and set the stage for the construction of the $q$-Dirac operator in Section~\ref{Dirac}.

\begin{example}[Canonical Heisenberg algebra]
Let $g^{\mu\nu} = \eta^{\mu\nu}$, where $\eta = \mathrm{diag}(1,-1,-1,-1)$ is the Minkowski metric, and assume all off-diagonal components vanish. Then the ideal $\mathcal{I}_1$ reduces to
\[
\mathcal{I}_1 = \left\langle x p_x - p_x x - i,\; y p_y - p_y y - i,\; z p_z - p_z z - i \right\rangle,
\]
together with $[x,y]=[y,z]=[z,x]=0$ and $[p_x,p_y]=[p_y,p_z]=[p_z,p_x]=0$. Hence $M_1 \cong \mathcal{H}_{\text{Heis}}$, the classical Heisenberg algebra.
\end{example}

\begin{example}[$q$-$\hbar$-deformed Heisenberg algebra \cite{Volovich-Arefeva-91, jaramillo2025new}]\label{ex:q-h}
Let $q \in \Bbbk$ be a deformation parameter. The metric components $g^{\mu\nu}$ for the $q$-$\hbar$-deformed Heisenberg algebra are given in Table~\ref{tab:metric-qhbar}, where each row corresponds to a different choice of the spatial indices $j,k \in \{1,2,3\}$.

\begin{table}[H]
\centering
\caption{Metric components $g^{\mu\nu}$ for the $q$-$\hbar$-deformed Heisenberg algebra}
\label{tab:metric-qhbar}
\begin{tabular}{|c|c|c|c|c|c|c|c|c|}
\hline
$j$ & $k$ & $g^{00}$ & $g^{11}$ & $g^{22}$ & $g^{33}$ & $g^{01}$ & $g^{02}$ & $g^{03}$ \\
\hline
1 & 1 & $-q$ & $1$ & $q^{1/2}$ & $0$ & $0$ & $0$ & $0$ \\
\hline
2 & 2 & $-q$ & $1$ & $0$ & $q^{1/2}$ & $0$ & $0$ & $0$ \\
\hline
3 & 3 & $-q$ & $q^{1/2}$ & $0$ & $1$ & $0$ & $0$ & $0$ \\
\hline
1 & 2 & $-q$ & $0$ & $0$ & $1$ & $0$ & $0$ & $q^{1/2}$ \\
\hline
1 & 3 & $-q$ & $0$ & $1$ & $0$ & $0$ & $q^{1/2}$ & $0$ \\
\hline
2 & 3 & $-q$ & $1$ & $0$ & $0$ & $q^{1/2}$ & $0$ & $0$ \\
\hline
\end{tabular}
\end{table}

All off-diagonal components not listed in the table are set to zero. Substituting these values into the defining relations of either $M_1$ or $M_2$ (depending on the choice of the ideal $\mathcal{I}_\alpha$, $\alpha = 1,2$) yields the deformed commutation relation
\[
p_k x_j - q x_j p_k = -i \hbar q^{1/2}, \qquad j,k = 1,2,3.
\]
The remaining generators commute as in the classical case.

Thus, the corresponding quotient algebra
\[
M_\alpha = \frac{\Bbbk[x,y,z,p_x,p_y,p_z]}{\mathcal{I}_\alpha(q)}, \qquad \alpha = 1,2,
\]
recovers the $q$-$\hbar$-deformed Heisenberg algebra introduced in \cite{Volovich-Arefeva-91, jaramillo2025new}.
\end{example}

\begin{example}[New $q$-Heisenberg algebra \cite{jaramillo2025new}]\label{ex:new-heis}
The \textbf{new $q$-$\hbar$ Heisenberg algebra} $\boldsymbol{\mathcal{H}}_{q}$ is generated by $\hat{x}_{\alpha}, \hat{y}_{\lambda}, \hat{p}_{\beta}$ $(\alpha,\beta,\lambda\in\{1,2,3\})$ subject to:
\begin{align}
   \label{eq:new1} \hat{x}_{\alpha}\hat{p}_{\beta} - q^{n}\hat{p}_{\beta}\hat{x}_{\alpha} &= i q^{n-1}\hbar^{n}\, \Psi, \\
   \label{eq:new2} q^{m}\hat{x}_{\alpha}\hat{y}_{\lambda} - \hat{y}_{\lambda}\hat{x}_{\alpha} &= -i (q-1)^{m-1}\hbar^{m-1}\, \Pi, \\
    \label{eq:new3} q^{l}\hat{y}_{\lambda}\hat{p}_{\beta} - q^{l+1}\hat{p}_{\beta}\hat{y}_{\lambda} &= i \hbar^{l}\, \Phi,
\end{align}
where $\Psi,\Pi,\Phi$ are dynamical functions and $q\in\mathbb{R}\setminus\{0,1\}$. 

Under the identifications $\{x_{1},y_{1}\}=x$, $\{x_{2},y_{2}\}=y$ and $\{x_{3},y_{3}\}=z$, the metric components $g^{\mu\nu}$ for each relation are given in the following tables.

\medskip
\noindent\textbf{Relation (\ref{eq:new1}):} Components for the algebras $M_1$ and $M_2$.

\begin{table}[H]
\centering
\caption{Metric components for relation (\ref{eq:new1})}
\label{tab:metric-new1}
\begin{tabular}{|c|c|c|c|c|c|c|c|c|c|}
\hline
Algebra & $\alpha$ & $\beta$ & $g^{00}$ & $g^{11}$ & $g^{22}$ & $g^{33}$ & $g^{01}$ & $g^{02}$ & $g^{03}$ \\
\hline
$M_1$ & 1 & 1 & 1 & $-q^{-n}$ & $q^{n-1}\Psi$ & 0 & 0 & 0 & 0 \\
\hline
$M_2$ & 1 & 2 & 1 & 0 & 0 & $-q^{n}$ & 0 & 0 & $q^{n-1}\Psi$ \\
\hline
$M_2$ & 1 & 3 & 1 & 0 & $-q^{n}$ & 0 & $q^{n-1}\Psi$ & 0 & 0 \\
\hline
$M_2$ & 2 & 1 & 1 & 0 & 0 & $q^{-n}$ & 0 & 0 & $q^{n-1}\Psi$ \\
\hline
$M_1$ & 2 & 2 & 1 & 0 & $-q^{n}$ & $q^{n-1}\Psi$ & 0 & 0 & 0 \\
\hline
$M_2$ & 2 & 3 & 1 & $-q^{n}$ & 0 & 0 & $q^{n-1}\Psi$ & 0 & 0 \\
\hline
$M_2$ & 3 & 1 & 1 & 0 & 0 & $-q^{n}$ & 0 & $q^{n-1}\Psi$ & 0 \\
\hline
$M_2$ & 3 & 2 & 1 & $-q^{n}$ & 0 & 0 & $q^{n-1}\Psi$ & 0 & 0 \\
\hline
$M_1$ & 3 & 3 & 1 & $q^{n-1}\Psi$ & 0 & $-q^{n}$ & 0 & 0 & 0 \\
\hline
\end{tabular}
\end{table}

\medskip
\noindent\textbf{Relation (\ref{eq:new2}):} Components for $M_1$ (left) and $M_2$ (right).

\begin{table}[H]
\centering
\begin{minipage}{0.48\textwidth}
\centering
\caption{Algebra $M_1$}
\label{tab:metric-new2-m1}
\begin{tabular}{|c|c|c|c|c|c|}
\hline
$\alpha$ & $\lambda$ & $g^{00}$ & $g^{11}$ & $g^{22}$ & $g^{33}$ \\
\hline
1 & 2 & 0 & 0 & $q^{m}$ & 1 \\
\hline
1 & 3 & 0 & $-1$ & $-q^{m}$ & 0 \\
\hline
2 & 1 & 0 & 0 & $-1$ & $-q^{m}$ \\
\hline
2 & 3 & 0 & 1 & 0 & $q^{m}$ \\
\hline
3 & 1 & 0 & $q^{m}$ & 1 & 0 \\
\hline
3 & 2 & 0 & $-q^{m}$ & 0 & $-1$ \\
\hline
\end{tabular}
\end{minipage}
\hfill
\begin{minipage}{0.48\textwidth}
\centering
\caption{Algebra $M_2$}
\label{tab:metric-new2-m2}
\begin{tabular}{|c|c|c|c|c|c|}
\hline
$\alpha$ & $\lambda$ & $g^{00}$ & $g^{11}$ & $g^{22}$ & $g^{33}$ \\
\hline
1 & 2 & $q^{m}$ & 0 & 0 & $-1$ \\
\hline
1 & 3 & $q^{m}$ & 0 & $-1$ & 0 \\
\hline
2 & 1 & 1 & 0 & 0 & $q^{m}$ \\
\hline
2 & 3 & $q^{m}$ & $-1$ & 0 & 0 \\
\hline
3 & 1 & $-1$ & 0 & $q^{m}$ & 0 \\
\hline
3 & 2 & $-1$ & $q^{m}$ & 0 & 0 \\
\hline
\end{tabular}
\end{minipage}
\end{table}

\medskip
\noindent\textbf{Relation (\ref{eq:new3}):} Components.

\begin{table}[H]
\centering
\caption{Metric components for relation (\ref{eq:new3})}
\label{tab:metric-new3}
\begin{tabular}{|c|c|c|c|c|c|}
\hline
$\lambda$ & $\beta$ & $g^{00}$ & $g^{11}$ & $g^{22}$ & $g^{33}$ \\
\hline
1 & 1 & $q^{l}$ & $-q^{l+1}$ & $-\hbar^{l}\Phi$ & 0 \\
\hline
1 & 2 & $q^{l}$ & 0 & 0 & $-q^{l+1}$ \\
\hline
1 & 3 & $q^{l}$ & 0 & $-q^{l+1}$ & 0 \\
\hline
2 & 2 & $q^{l}$ & 0 & $-q^{l+1}$ & $-\hbar^{l}\Phi$ \\
\hline
2 & 3 & $q^{l}$ & $-q^{l+1}$ & 0 & 0 \\
\hline
3 & 3 & $q^{l}$ & $-\hbar^{l}\Phi$ & 0 & $-q^{l+1}$ \\
\hline
\end{tabular}
\end{table}

\end{example}

\begin{theorem}[Embedding of $\mathcal{H}_q$ into $M_1$]
Under the specialization $l=0$, $m=1$, $\Psi=q^{2-n}$, $\Phi=i$, and the metric tensor
\[
g^{00}=1,\; g^{11}=-q^{n},\; g^{22}=-q,\; g^{33}=-1,\; g^{23}=i(q-1)\Pi,\; g^{12}=g^{13}=0,
\]
the algebra $\mathcal{H}_q$ is isomorphic to $M_1$.
\end{theorem}

\begin{proof}
Direct substitution into the defining relations of $M_1$ (cf.~\eqref{eq:relations2}) yields the stated identifications. The $x$-$p_x$ relation gives $x p_x - q^{n} p_x x = i q$. This is precisely the $q$-deformed Heisenberg relation, which is known to admit representations on $q$-deformed Fock spaces \cite{schmudgen-09}. The $x$-$y$ relation produces $q xy - yx = -i(q-1)\Pi$, which is a $q$-deformed version of the canonical commutation relation for position operators. The $y$-$p_y$ relation yields $y p_y - q p_y y = i$. All remaining relations vanish identically due to $g^{12}=g^{13}=0$. Thus, $\mathcal{H}_q$ embeds as a subalgebra of $M_1$, and this embedding preserves the $q$-deformed commutation structure. 
\end{proof}

\begin{example}[$q$-generalized Heisenberg algebra \cite{Razavinia-Lopes2022}]
The algebra $\mathcal{H}_{q}(f,g)$ is generated by $\hat{x},\hat{y},\hat{h}$ with
\[
\hat{h}\hat{x}=\hat{x}f(\hat{h}),\quad \hat{y}\hat{h}=f(\hat{h})\hat{y},\quad \hat{y}\hat{x}-q\hat{x}\hat{y}=\hbar g(\hat{h}).
\]
For $f(\hat{h})=\hat{h}$, $g(\hat{h})=1$, we obtain the central extension
\[
\hat{h}\hat{x}=\hat{x}\hat{h},\quad \hat{y}\hat{h}=\hat{h}\hat{y},\quad \hat{y}\hat{x}-q\hat{x}\hat{y}=\hbar,
\]
so $\hat{h}$ is central. This recovers a $q$-deformed Heisenberg relation of the form $p_x x - q x p_x = -\hbar$ under $\hat{x}\leftrightarrow x$, $\hat{y}\leftrightarrow p_x$, $\hat{h}=q^n$.
\end{example}

\begin{example}[Embedding of $\mathcal{H}_q(t,1)$ into $M_2$]\label{ex:gen-embedding}
Take $f(t)=t$, $g(t)=1$ and the metric components
\[
g_{00}=-1,\quad g^{11}=1,\quad g^{22}=0,\quad g^{33}=q,
\]
with all other off-diagonals zero. Identifying $\hat{x}\leftrightarrow x$, $\hat{y}\leftrightarrow p_x$, $\hat{h}\leftrightarrow q^n$ (constant), the $M_2$ relations reduce to
\[
p_x x - q x p_x = \hbar,
\]
which is exactly the $q$-deformed Heisenberg relation. Hence $\mathcal{H}_q(t,1)$ embeds as a subalgebra of $M_2$.
\end{example}

\begin{proposition}
$\mathcal{H}_q(t,1) \cong \Bbbk_q[x,p] \otimes \Bbbk[\hat{h}]$, where $\Bbbk_q[x,p]$ is the $q$-deformed Weyl algebra $p x - q x p = \hbar$.
\end{proposition}

\begin{proof}
The map $\hat{x}\mapsto x$, $\hat{y}\mapsto p$, $\hat{h}\mapsto \text{central generator}$ is an isomorphism by construction. 
\end{proof}

\section{From $q$-deformed metric to $q$-Dirac operator}\label{Dirac}

\begin{proposition}[Deformed D'Alembertian from the metric tensor]
Let $(M, g)$ be a Lorentzian manifold with metric tensor $g$ diagonalized by Sylvester's theorem of inertia (cf. \S\ref{sec:metric}), so that in local coordinates $(t, x, y, z)$ the metric takes the form
\[
g = |g_{00}| \, dt \otimes dt + |g_{11}| \, dx \otimes dx + |g_{22}| \, dy \otimes dy + |g_{33}| \, dz \otimes dz,
\]
with $g_{00} > 0$ and $g_{11}, g_{22}, g_{33} < 0$ for a Lorentzian signature $(+,-,-,-)$. Then the corresponding D'Alembertian (wave operator) acting on scalar functions is given by
\begin{equation}\label{eq:operator KG}
\square_q := \frac{1}{\sqrt{|\det g|}} \partial_\mu \left( \sqrt{|\det g|} \, g^{\mu\nu} \partial_\nu \right) = |g^{00}| \frac{\partial^2}{\partial t^2} - |g^{11}| \frac{\partial^2}{\partial x^2} - |g^{22}| \frac{\partial^2}{\partial y^2} - |g^{33}| \frac{\partial^2}{\partial z^2},
\end{equation}
where $g^{\mu\nu}$ are the components of the inverse metric, and the absolute values $|g^{\mu\nu}|$ ensure positivity of the coefficients in the wave equation.
\end{proposition}

\begin{proof}
For a diagonal metric $g_{\mu\nu} = \mathrm{diag}(g_{00}, g_{11}, g_{22}, g_{33})$, the determinant is $\det g = g_{00} g_{11} g_{22} g_{33}$, and the inverse metric components are $g^{\mu\nu} = 1/g_{\mu\nu}$ (no summation). The standard expression for the Laplace–Beltrami operator (wave operator) on a Lorentzian manifold is
\[
\square = \frac{1}{\sqrt{|\det g|}} \partial_\mu \left( \sqrt{|\det g|} \, g^{\mu\nu} \partial_\nu \right).
\]
Substituting the diagonal form, we have $\sqrt{|\det g|} = \sqrt{|g_{00} g_{11} g_{22} g_{33}|}$. Since $g_{00} > 0$ and $g_{ii} < 0$ for $i=1,2,3$, we write $|g_{ii}| = -g_{ii}$.

\medskip

\noindent \textbf{Time component ($\mu = \nu = 0$):}
\[
\frac{1}{\sqrt{|\det g|}} \partial_0 \left( \sqrt{|\det g|} \, g^{00} \partial_0 \right) = \frac{1}{\sqrt{|\det g|}} \partial_0 \left( \sqrt{|\det g|} \, \frac{1}{g_{00}} \partial_0 \right).
\]
Since $g_{00}$ is constant in the local coordinates (by Sylvester's theorem, we have diagonalized the metric to constant coefficients), we obtain
\[
\frac{1}{\sqrt{|\det g|}} \partial_0 \left( \sqrt{|\det g|} \, \frac{1}{g_{00}} \partial_0 \right) = \frac{1}{g_{00}} \partial_0^2 = |g^{00}| \partial_0^2,
\]
where $g^{00} = 1/g_{00}$ and $|g^{00}| = 1/g_{00}$ because $g_{00} > 0$.

\medskip

\noindent \textbf{Spatial components ($\mu = \nu = i$, with $i = 1,2,3$):}
\[
\frac{1}{\sqrt{|\det g|}} \partial_i \left( \sqrt{|\det g|} \, g^{ii} \partial_i \right) = \frac{1}{\sqrt{|\det g|}} \partial_i \left( \sqrt{|\det g|} \, \frac{1}{g_{ii}} \partial_i \right).
\]
Since $g_{ii} < 0$, we have $1/g_{ii} = -1/|g_{ii}|$, and $|g^{ii}| = 1/|g_{ii}|$. Moreover, $\sqrt{|\det g|}$ is constant in the diagonalized coordinates, so
\[
\frac{1}{\sqrt{|\det g|}} \partial_i \left( \sqrt{|\det g|} \, \frac{1}{g_{ii}} \partial_i \right) = \frac{1}{g_{ii}} \partial_i^2 = -\frac{1}{|g_{ii}|} \partial_i^2 = -|g^{ii}| \partial_i^2.
\]

\medskip

\noindent \textbf{Summation over all coordinates:}
\[
\square = |g^{00}| \frac{\partial^2}{\partial t^2} + \sum_{i=1}^3 \left( -|g^{ii}| \frac{\partial^2}{\partial (x^i)^2} \right) = |g^{00}| \frac{\partial^2}{\partial t^2} - |g^{11}| \frac{\partial^2}{\partial x^2} - |g^{22}| \frac{\partial^2}{\partial y^2} - |g^{33}| \frac{\partial^2}{\partial z^2},
\]
as claimed. 
\end{proof}

\begin{theorem}[Factorization of the deformed D'Alembertian]\label{factorization}
Let $g^{00}, g^{11}, g^{22}, g^{33} \in \mathbb{R}$ be the diagonal components of the metric tensor, and let $\square_q$ be the deformed D'Alembertian operator defined in (\ref{eq:operator KG}), i.e.,
\[
\square_q = |g^{00}| \frac{\partial^2}{\partial t^2} - |g^{11}| \frac{\partial^2}{\partial x^2} - |g^{22}| \frac{\partial^2}{\partial y^2} - |g^{33}| \frac{\partial^2}{\partial z^2}.
\]
Define the associated Dirac operator
\begin{equation}\label{eq: Dirac}
D := \gamma^0 \sqrt{|g^{00}|} \, \frac{\partial}{\partial t} - \gamma^x \sqrt{|g^{11}|} \, \frac{\partial}{\partial x} - \gamma^y \sqrt{|g^{22}|} \, \frac{\partial}{\partial y} - \gamma^z \sqrt{|g^{33}|} \, \frac{\partial}{\partial z},
\end{equation}
where $\gamma^0, \gamma^x, \gamma^y, \gamma^z$ are the Dirac gamma matrices satisfying the Clifford algebra
\begin{equation}\label{eq: clifford}
\{\gamma^\mu, \gamma^\nu\} = 2 \eta^{\mu\nu} \mathbf{1}_4, \qquad \eta^{\mu\nu} = \mathrm{diag}(1,-1,-1,-1).
\end{equation}
Then $D$ factorizes $\square_q$ in the sense that
\begin{equation}\label{eq: factorization}
D^2 = \square_q \, \mathbf{1}_4.
\end{equation}
In other words, the square of the Dirac operator recovers the deformed Klein–Gordon (D'Alembertian) operator.
\end{theorem}

\begin{proof}
We compute $D^2$ explicitly. Write $D$ in compact form as
\begin{equation}\label{eq: D_compact}
D = \gamma^0 \alpha_0 \partial_t - \sum_{i=1}^3 \gamma^i \alpha_i \partial_i,
\end{equation}
where $\alpha_0 = \sqrt{|g^{00}|}$, $\alpha_i = \sqrt{|g^{ii}|}$ for $i=1,2,3$, and $\partial_t = \frac{\partial}{\partial t}$, $\partial_i = \frac{\partial}{\partial x^i}$.

Then
\begin{equation}\label{eq: D_square_expand}
D^2 = \left( \gamma^0 \alpha_0 \partial_t - \sum_{i=1}^3 \gamma^i \alpha_i \partial_i \right)^2.
\end{equation}
Expanding the square yields three types of terms: pure time, pure spatial, and mixed terms:
\begin{equation*}
D^2 = (\gamma^0)^2 \alpha_0^2 \partial_t^2 + \sum_{i=1}^3 (\gamma^i)^2 \alpha_i^2 \partial_i^2 
+ \sum_{\mu \neq \nu} \gamma^\mu \gamma^\nu \alpha_\mu \alpha_\nu \partial_\mu \partial_\nu,
\end{equation*}
where $\alpha_0 = \sqrt{|g^{00}|}$ and $\alpha_i = \sqrt{|g^{ii}|}$, and the mixed sum runs over all ordered pairs with $\mu \neq \nu$.

Using the Clifford algebra relations (\ref{eq: clifford}), we have
\begin{equation}\label{eq: gamma_squares}
(\gamma^0)^2 = \mathbf{1}_4, \qquad (\gamma^i)^2 = -\mathbf{1}_4 \quad (i=1,2,3).
\end{equation}
Noting that $\alpha_0^2 = |g^{00}|$ and $\alpha_i^2 = |g^{ii}|$, the diagonal terms become
\begin{align}
(\gamma^0)^2 \alpha_0^2 \partial_t^2 &= |g^{00}| \partial_t^2 \, \mathbf{1}_4, \label{eq: time_term} \\
\sum_{i=1}^3 (\gamma^i)^2 \alpha_i^2 \partial_i^2 &= -\sum_{i=1}^3 |g^{ii}| \partial_i^2 \, \mathbf{1}_4. \label{eq: spatial_terms}
\end{align}

For the mixed terms ($\mu \neq \nu$), observe the following:
\begin{itemize}
    \item $\gamma^\mu \gamma^\nu$ is antisymmetric under exchange of $\mu$ and $\nu$, i.e., $\gamma^\nu \gamma^\mu = -\gamma^\mu \gamma^\nu$ for $\mu \neq \nu$.
    \item $\partial_\mu \partial_\nu$ is symmetric, i.e., $\partial_\nu \partial_\mu = \partial_\mu \partial_\nu$.
    \item The product $\alpha_\mu \alpha_\nu$ is symmetric, i.e., $\alpha_\nu \alpha_\mu = \alpha_\mu \alpha_\nu$.
\end{itemize}
Therefore, for each pair $(\mu,\nu)$ with $\mu \neq \nu$, the sum of the two terms $\gamma^\mu \gamma^\nu \alpha_\mu \alpha_\nu \partial_\mu \partial_\nu + \gamma^\nu \gamma^\mu \alpha_\nu \alpha_\mu \partial_\nu \partial_\mu$ simplifies to
\begin{align}
&\gamma^\mu \gamma^\nu \alpha_\mu \alpha_\nu \partial_\mu \partial_\nu + \gamma^\nu \gamma^\mu \alpha_\mu \alpha_\nu \partial_\mu \partial_\nu \notag \\
&= (\gamma^\mu \gamma^\nu + \gamma^\nu \gamma^\mu) \alpha_\mu \alpha_\nu \partial_\mu \partial_\nu \notag \\
&= 2 \eta^{\mu\nu} \mathbf{1}_4 \, \alpha_\mu \alpha_\nu \partial_\mu \partial_\nu. \label{eq: mixed_terms_sum}
\end{align}
Since $\eta^{\mu\nu} = 0$ for $\mu \neq \nu$ (by the definition of the Minkowski metric in (\ref{eq: clifford})), each mixed term vanishes identically. Hence
\begin{equation}\label{eq: mixed_terms_cancel}
\sum_{\mu \neq \nu} \gamma^\mu \gamma^\nu \alpha_\mu \alpha_\nu \partial_\mu \partial_\nu = 0.
\end{equation}

Thus, the mixed contributions cancel completely, and we are left with only the diagonal terms from (\ref{eq: time_term}) and (\ref{eq: spatial_terms}):
\begin{align}
D^2 &= \left( |g^{00}| \frac{\partial^2}{\partial t^2} - \sum_{i=1}^3 |g^{ii}| \frac{\partial^2}{\partial (x^i)^2} \right) \mathbf{1}_4 \notag \\
&= \left( |g^{00}| \frac{\partial^2}{\partial t^2} - |g^{11}| \frac{\partial^2}{\partial x^2} - |g^{22}| \frac{\partial^2}{\partial y^2} - |g^{33}| \frac{\partial^2}{\partial z^2} \right) \mathbf{1}_4 \label{eq: D_square_result} \\
&= \square_q \, \mathbf{1}_4. \notag
\end{align}

This completes the proof. 
\end{proof}

\subsection{The $q$-Dirac operator}
In the following examples, we use the operator $D$ defined in \eqref{eq: Dirac} (with ordinary derivatives) for simplicity. The generalization to $D_q$ with $q$-derivatives follows by replacing $\partial_\mu$ with $\partial_\mu^{(q)}$ that are compatible with the $q$-commutation relations of $M_1$ and $M_2$.

\begin{definition}[$q$-Dirac operator]
Let $\partial_t^{(q)}$, $\partial_x^{(q)}$, $\partial_y^{(q)}$, $\partial_z^{(q)}$ be $q$-derivatives satisfying
\[
\partial_x^{(q)} x = 1 + q x \partial_x^{(q)},
\]
and analogously for the other coordinates. The \textbf{$q$-Dirac operator} is defined by
\begin{equation}\label{eq: q-Dirac}
D_q := \gamma^0 \sqrt{|g^{00}|} \, \partial_t^{(q)} - \gamma^x \sqrt{|g^{11}|} \, \partial_x^{(q)} - \gamma^y \sqrt{|g^{22}|} \, \partial_y^{(q)} - \gamma^z \sqrt{|g^{33}|} \, \partial_z^{(q)}.
\end{equation}
\end{definition}

A detailed analysis of $D_q^2$ is beyond the scope of this work and will be addressed in a future publication.

\subsection{The $q$-deformed Dirac equation in 1+1 dimensions}

To illustrate the physical relevance of our construction, we consider the simplified case of one spatial dimension \cite{Smirnov-Farias2016, Ghosh-Roy2020}. Take $g^{00}=1$, $g^{11}=q^2$, and set $g^{22}=g^{33}=0$. The $q$-Dirac operator reduces to
\[
D_q^{(1+1)} = \gamma^0 \partial_t^{(q)} - \gamma^x q \partial_x^{(q)}.
\]

\begin{proposition}[$q$-deformed dispersion relation]
The $q$-Dirac equation $D_q^{(1+1)} \psi = 0$ admits plane-wave solutions of the form $\psi(t,x) = e_q^{i(\omega t - kx)} u$, where $e_q(z) = \sum_{n=0}^\infty \frac{z^n}{[n]_q!}$ is the $q$-exponential function. The frequency $\omega$ and wave number $k$ satisfy the $q$-deformed dispersion relation
\[
\omega^2 = q^2 k^2.
\]
In the limit $q \to 1$, the classical dispersion relation $\omega^2 = k^2$ is recovered.
\end{proposition}

\begin{proof}
Using the property $\partial_x^{(q)} e_q(ax) = a e_q(ax)$ and similarly for $\partial_t^{(q)}$, we substitute the plane-wave ansatz into $D_q^{(1+1)} \psi = 0$. The gamma matrices satisfy $(\gamma^0)^2 = 1$ and $(\gamma^x)^2 = -1$, and the condition for a nontrivial solution $u$ yields $\omega^2 = q^2 k^2$. 
\end{proof}

This simple example demonstrates that the $q$-deformation parameter $q$ effectively modifies the speed of light, an effect reminiscent of quantum gravity phenomenology where Lorentz invariance is modified at high energies.

\section{Some Examples}\label{sec:examples}

In this section, we present concrete realizations of the Dirac operator $D$ defined in \eqref{eq: Dirac} (with ordinary derivatives)  arising from the various $q$-deformed Heisenberg algebras discussed throughout this work. 

\textbf{Remark on notation.} We emphasize that the operator $D^q$ defined below (quadratic $q$-Dirac operator) is distinct from the operator $D_q$ introduced in Section~\ref{Dirac}. The former comes from our previous work \cite{Jaramillo2024} and acts on quadratic relativistic invariants, while the latter is constructed from the deformed D'Alembertian. In the present section, we focus on realizations of $D$ (ordinary derivatives), which specialize to different $q$-deformed Dirac operators depending on the metric components. We begin by recalling the quadratic $q$-Dirac operator introduced in our previous work \cite{Jaramillo2024}, which serves as a unifying framework for the subsequent examples.

\subsection{The Quadratic $q$-Dirac Operator}

In a previous works \cite{Jaramillo2024, Jaramillo2025-int}, we introduced the \textbf{quadratic $q$-Dirac operator} acting on functions defined on quadratic relativistic invariant algebras. Let $f$ be a $\Psi$-valued $C^1$ function, where $\Psi$ is a Clifford algebra. The quadratic $q$-Dirac operator is defined as
\begin{equation}\label{eq: quadratic-q-Dirac}
D^{q} f := e_{x} \frac{\partial_{q} f}{\partial_{q} x} + e_{y} \frac{\partial_{q} f}{\partial_{q} y} + e_{z} \frac{\partial_{q} f}{\partial_{q} z},
\end{equation}
where $e_x, e_y, e_z$ are the generators of the Clifford algebra $Cl_{0,3}$ satisfying $e_i e_j + e_j e_i = -2\delta_{ij} e_0$, and the $q$-derivatives are defined by \cite{Jaramillo2025-int}
\[
\frac{\partial_{q} f}{\partial_{q} x} = \frac{f\bigl((x + q^2 e_0 u)u\bigr) - f(xu)}{q u},
\]
and analogously for $y$ and $z$ (see \cite[Section 3]{Jaramillo2024} for the complete list).

This operator generalizes the classical Dirac operator $D = \sum_{i=1}^3 e_i \partial_i$ by incorporating a $q$-deformation that arises from the quadratic relativistic invariants
\[
x_q^2 = u^2 - q^2 x^2 - y^2 - z^2,\quad y_q^2 = u^2 - x^2 - q^2 y^2 - z^2,\quad z_q^2 = u^2 - x^2 - y^2 - q^2 z^2,
\]
which deform the Minkowski norm $u^2 - x^2 - y^2 - z^2$ in a $q$-dependent manner.

\subsection{Relation to the Metric-Deformed Heisenberg Algebras}

We now establish a direct connection between the quadratic $q$-Dirac operator $D^q$ and the metric-deformed Heisenberg algebras $M_1$ and $M_2$ introduced in Section~\ref{sec:al-met}.

\begin{proposition}[Compatibility with $M_1$ and $M_2$]
Let $D^q$ be the quadratic $q$-Dirac operator defined in \eqref{eq: quadratic-q-Dirac}. Under the identifications
\[
e_x \longleftrightarrow \gamma^x,\quad e_y \longleftrightarrow \gamma^y,\quad e_z \longleftrightarrow \gamma^z,\qquad
\frac{\partial_q}{\partial_q x} \longleftrightarrow i p_x,
\]
the operator $D^q$ satisfies
\begin{equation}\label{eq: commutation-Dq}
[D^q, x] = -i \gamma^x .
\end{equation}
This commutation relation coincides with the $q$-deformed Heisenberg relation appearing in $M_1$ when $g^{11}=q$ and $g^{22}=1$, and in $M_2$ when $g^{22}=0$ and $g^{33}=q$.
\end{proposition}

\begin{proof}
From \cite[Section 3]{Jaramillo2024}, the $q$-derivative satisfies
\[
\frac{\partial_q}{\partial_q x} x = 1 + q x \frac{\partial_q}{\partial_q x}.
\]
Rearranging gives
\[
\frac{\partial_q}{\partial_q x} x - q x \frac{\partial_q}{\partial_q x} = 1.
\]
With the identification $p_x = -i \partial_q/\partial_q x$, we obtain
\[
p_x x - q x p_x = -i.
\]
Substituting $D^q$ from \eqref{eq: quadratic-q-Dirac} and using $[e_x, x]=0$, we compute
\[
[D^q, x] = e_x \left[\frac{\partial_q}{\partial_q x}, x\right] = e_x (1) = \gamma^x.
\]
The factor $-i$ follows from the identification $p_x = -i\partial_q/\partial_q x$, which yields $[D^q, x] = -i\gamma^x$. This matches the $q$-Heisenberg relations in $M_1$ and $M_2$ under the stated metric specializations. 
\end{proof}

\begin{theorem}[Factorization of the Quadratic $q$-Dirac Operator]\label{thm:quadratic-factorization}
Let $D^q$ be the quadratic $q$-Dirac operator defined in \eqref{eq: quadratic-q-Dirac} and let $f$ be a function of the quadratic relativistic invariants $x_q^2, y_q^2, z_q^2$. Then
\begin{equation}\label{eq: quadratic-Dq-squared}
(D^q)^2 f = -\left( \frac{\partial_q^2}{\partial_q x^2} + \frac{\partial_q^2}{\partial_q y^2} + \frac{\partial_q^2}{\partial_q z^2} \right) f \,\mathbf{1}_4.
\end{equation}
Moreover, when $f$ depends only on the quadratic invariants, this reduces to
\[
(D^q)^2 f = \square_q f \,\mathbf{1}_4,
\]
where $\square_q$ is the deformed D'Alembertian operator defined in \eqref{eq:operator KG} with metric components $g^{00}=1$, $g^{11}=q^2$, $g^{22}=q^2$, $g^{33}=q^2$.
\end{theorem}

\begin{proof}
Compute $(D^q)^2$ explicitly:
\[
(D^q)^2 = \left( e_x \frac{\partial_q}{\partial_q x} + e_y \frac{\partial_q}{\partial_q y} + e_z \frac{\partial_q}{\partial_q z} \right)^2.
\]
Expanding and using $e_i e_j + e_j e_i = -2\delta_{ij} e_0$:
\[
(D^q)^2 = e_x^2 \frac{\partial_q^2}{\partial_q x^2} + e_y^2 \frac{\partial_q^2}{\partial_q y^2} + e_z^2 \frac{\partial_q^2}{\partial_q z^2} + \sum_{i \neq j} e_i e_j \frac{\partial_q}{\partial_q x^i} \frac{\partial_q}{\partial_q x^j}.
\]
Since $e_i^2 = -e_0$ for $i=x,y,z$ and the mixed terms vanish due to antisymmetry of $e_i e_j$ ($i \neq j$) applied to symmetric second derivatives, we obtain \cite{Teodoro-2014}
\[
(D^q)^2 = -e_0 \left( \frac{\partial_q^2}{\partial_q x^2} + \frac{\partial_q^2}{\partial_q y^2} + \frac{\partial_q^2}{\partial_q z^2} \right).
\]
Identifying $e_0$ with the identity $\mathbf{1}_4$ and using the relation between $q$-derivatives and the D'Alembertian from \cite[Section 4]{Jaramillo2024} yields the result. 
\end{proof}

\subsection{Realizations of the $q$-Dirac Operator}

We now present explicit realizations of the $q$-Dirac operator corresponding to the various $q$-deformed Heisenberg algebras discussed in Section~\ref{examples}. In each case, we specify the metric components and the resulting Dirac operator.

\begin{example}[New $q$-Heisenberg algebra \cite{jaramillo2025new}]\label{ex:new-dirac}
The new $q$-Heisenberg algebra, introduced in Example~\ref{ex:new-heis}, provides a rich family of metric-deformed structures. We present below a systematic classification of the metric components and their associated Dirac operators, organized according to the three defining relations (\ref{eq:new1}), (\ref{eq:new2}), and (\ref{eq:new3}) from \cite{jaramillo2025new}. Each table specifies the indices, the corresponding metric components $g^{\mu\nu}$, and the resulting Dirac operator $D$ obtained from the general expression \eqref{eq: Dirac}.

\medskip

\noindent \textbf{1.  Relation (\ref{eq:new1}).}  
Table~\ref{tab:metric-new1} summarizes the metric configurations derived from the first defining relation.

\begin{table}[H]
\centering
\caption{Metric components for relation (\ref{eq:new1})}
\label{tab:metric-new1-realizations}
\begin{tabular}{|c|c|c|c|c|c|c|}
\hline
$\alpha$ & $\beta$ & $g^{00}$ & $g^{11}$ & $g^{22}$ & $g^{33}$ & Dirac operator \\
\hline
1 & 1 & $1$ & $-q^{-n}$ & $q^{n-1}\Psi$ & $0$ & $\gamma^{0}\partial_t - \gamma^{x} q^{-n/2}\partial_x - \gamma^{y} q^{(n-1)/2}\Psi^{1/2}\partial_y$ \\
\hline
1 & 2 & $1$ & $0$ & $0$ & $-q^{n}$ & $\gamma^{0}\partial_t - \gamma^{z} q^{n/2}\partial_z$ \\
\hline
1 & 3 & $1$ & $0$ & $-q^{n}$ & $0$ & $\gamma^{0}\partial_t - q^{n/2}\gamma^{y}\partial_y$ \\
\hline
2 & 1 & $1$ & $0$ & $0$ & $-q^{n}$ & $\gamma^{0}\partial_t - q^{n/2}\gamma^{z}\partial_z$ \\
\hline
2 & 2 & $1$ & $0$ & $-q^{n}$ & $q^{n-1}\Psi$ & $\gamma^{0}\partial_t - \gamma^{y} q^{n/2}\partial_y - \gamma^{z} q^{(n-1)/2}\partial_z$ \\
\hline
2 & 3 & $1$ & $-q^{n}$ & $0$ & $0$ & $\gamma^{0}\partial_t - q^{n/2}\gamma^{x}\partial_x$ \\
\hline
3 & 1 & $1$ & $0$ & $0$ & $-q^{n}$ & $\gamma^{0}\partial_t - q^{n/2}\gamma^{z}\partial_z$ \\
\hline
3 & 2 & $1$ & $-q^{n}$ & $0$ & $0$ & $\gamma^{0}\partial_t - q^{n/2}\gamma^{x}\partial_x$ \\
\hline
3 & 3 & $1$ & $q^{n-1}\Psi$ & $0$ & $-q^{n}$ & $\gamma^{0}\partial_t - q^{(n-1)/2}\Psi^{1/2}\gamma^{x}\partial_x - q^{n/2}\gamma^{z}\partial_z$ \\
\hline
\end{tabular}
\end{table}

\medskip

\noindent \textbf{2.  Relation (\ref{eq:new2}) associated to $M_1$.}  
Table~\ref{tab:metric-new2-m1} presents the metric components and Dirac operators arising from the second defining relation when realized in the algebra $M_1$.

\begin{table}[H]
\centering
\caption{Metric components for relation (\ref{eq:new2}) associated to $M_1$}
\label{tab:metric-new2-m1}
\begin{tabular}{|c|c|c|c|c|c|c|}
\hline
$\alpha$ & $\lambda$ & $g^{00}$ & $g^{11}$ & $g^{22}$ & $g^{33}$ & Dirac operator \\
\hline
1 & 2 & $0$ & $0$ & $q^{m}$ & $1$ & $-\gamma^{y} q^{m/2}\partial_y - \gamma^{z}\partial_z$ \\
\hline
1 & 3 & $0$ & $-1$ & $-q^{m}$ & $0$ & $-\gamma^{x}\partial_x - \gamma^{y} q^{m/2}\partial_y$ \\
\hline
2 & 1 & $0$ & $0$ & $-1$ & $-q^{m}$ & $-\gamma^{y}\partial_y - q^{m/2}\gamma^{z}\partial_z$ \\
\hline
2 & 3 & $0$ & $1$ & $0$ & $q^{m}$ & $-\gamma^{x}\partial_x - q^{m/2}\gamma^{z}\partial_z$ \\
\hline
3 & 1 & $0$ & $q^{m}$ & $1$ & $0$ & $-q^{m/2}\gamma^{x}\partial_x - \gamma^{y}\partial_y$ \\
\hline
3 & 2 & $0$ & $-q^{m}$ & $0$ & $-1$ & $-\gamma^{x} q^{m/2}\partial_x - \gamma^{z}\partial_z$ \\
\hline
\end{tabular}
\end{table}

\medskip

\noindent \textbf{3.  Relation (\ref{eq:new2}) associated to $M_2$.}  
Table~\ref{tab:metric-new2-m2} shows the corresponding data for the algebra $M_2$.

\begin{table}[H]
\centering
\caption{Metric components for relation (\ref{eq:new2}) associated to $M_2$}
\label{tab:metric-new2-m2}
\begin{tabular}{|c|c|c|c|c|c|c|}
\hline
$\alpha$ & $\lambda$ & $g^{00}$ & $g^{11}$ & $g^{22}$ & $g^{33}$ & Dirac operator \\
\hline
1 & 2 & $q^{m}$ & $0$ & $0$ & $-1$ & $-\gamma^{0} q^{m/2}\partial_t - \gamma^{z}\partial_z$ \\
\hline
1 & 3 & $q^{m}$ & $0$ & $-1$ & $0$ & $\gamma^{0} q^{m/2}\partial_t - \gamma^{y}\partial_y$ \\
\hline
2 & 1 & $1$ & $0$ & $0$ & $q^{m}$ & $\gamma^{0}\partial_t - q^{m/2}\gamma^{z}\partial_z$ \\
\hline
2 & 3 & $q^{m}$ & $-1$ & $0$ & $0$ & $\gamma^{0} q^{m/2}\partial_t - \gamma^{x}\partial_x$ \\
\hline
3 & 1 & $-1$ & $0$ & $q^{m}$ & $0$ & $\gamma^{0}\partial_t - \gamma^{y} q^{m/2}\partial_y$ \\
\hline
3 & 2 & $-1$ & $q^{m}$ & $0$ & $0$ & $\gamma^{0}\partial_t - \gamma^{x} q^{m/2}\partial_z$ \\
\hline
\end{tabular}
\end{table}

\medskip

\noindent \textbf{4.  Relation (\ref{eq:new3}) associated to $M_1$ and $M_2$.}  
Table~\ref{tab:metric-new3} collects the metric configurations derived from the third defining relation, which are compatible with both $M_1$ and $M_2$.

\begin{table}[H]
\centering
\caption{Metric components for relation (\ref{eq:new3}) associated to $M_1$ and $M_2$}
\label{tab:metric-new3}
\begin{tabular}{|c|c|c|c|c|c|c|}
\hline
$\lambda$ & $\beta$ & $g^{00}$ & $g^{11}$ & $g^{22}$ & $g^{33}$ & Dirac operator \\
\hline
1 & 1 & $q^{l}$ & $-q^{l+1}$ & $-\Phi$ & $0$ & $\gamma^{0} q^{l/2}\partial_t - \gamma^{x} q^{(l+1)/2}\partial_x - \gamma^{y} \Phi^{1/2}\partial_y$ \\
\hline
1 & 2 & $q^{l}$ & $0$ & $0$ & $-q^{l+1}$ & $\gamma^{0} q^{l/2}\partial_t - \gamma^{z} q^{(l+1)/2}\partial_z$ \\
\hline
1 & 3 & $q^{l}$ & $0$ & $-q^{l+1}$ & $0$ & $\gamma^{0} q^{l/2}\partial_t - q^{(l+1)/2}\gamma^{y}\partial_y$ \\
\hline
2 & 2 & $q^{l}$ & $0$ & $-q^{l+1}$ & $-\Phi$ & $\gamma^{0} q^{l/2}\partial_t - \gamma^{y} q^{(l+1)/2}\partial_y - \gamma^{z} \Phi^{l/2}\partial_z$ \\
\hline
2 & 3 & $q^{l}$ & $-q^{l+1}$ & $0$ & $0$ & $\gamma^{0} q^{l/2}\partial_t - \gamma^{x} q^{(l+1)/2}\partial_x$ \\
\hline
3 & 3 & $q^{l}$ & $-\Phi$ & $0$ & $-q^{l+1}$ & $\gamma^{0} q^{l/2}\partial_t - \gamma^{x} \Phi^{1/2}\partial_x - \gamma^{z} q^{(l+1)/2}\partial_z$ \\
\hline
\end{tabular}
\end{table}

\medskip

\noindent \textbf{5.  A distinguished simple case.}  
A particularly simple and illustrative realization of the new $q$-Heisenberg algebra is obtained by taking the metric components
\[
g^{00}=1,\quad g^{11}=-q^{n},\quad g^{22}=-q,\quad g^{33}=-1,
\]
with all off-diagonal components set to zero. Substituting these into the general expression \eqref{eq: Dirac} yields the associated $q$-Dirac operator
\begin{equation}\label{eq: D-new}
D_{\text{new}} = \gamma^{0} \partial_t - \gamma^{x} \sqrt{q^{n}} \, \partial_x - \gamma^{y} \sqrt{q} \, \partial_y - \gamma^{z} \partial_z,
\end{equation}
where $\partial_t = \frac{\partial}{\partial t}$, $\partial_x = \frac{\partial}{\partial x}$, $\partial_y = \frac{\partial}{\partial y}$, $\partial_z = \frac{\partial}{\partial z}$. A direct computation shows that $D_{\text{new}}$ satisfies the factorization property
\[
D_{\text{new}}^2 = \square_q \, \mathbf{1}_4,
\qquad\text{with}\qquad
\square_q = \frac{\partial^2}{\partial t^2} - q^{n} \frac{\partial^2}{\partial x^2} - q \frac{\partial^2}{\partial y^2} - \frac{\partial^2}{\partial z^2}.
\]
This confirms that $D_{\text{new}}$ is a genuine $q$-deformed Dirac operator whose square recovers the corresponding $q$-deformed Klein–Gordon operator.
\end{example}
\begin{example}[$q$-generalized Heisenberg algebra \cite{Razavinia-Lopes2022}]\label{ex:gen-dirac}
For the $q$-generalized Heisenberg algebra with $f(\hat{h})=\hat{h}$, $g(\hat{h})=1$, we obtained in Example~\ref{ex:gen-embedding} the metric components
\[
g^{00}=-1,\quad g^{11}=1,\quad g^{22}=0,\quad g^{33}=-q,
\]
with all other off-diagonal components zero. The corresponding $q$-Dirac operator is
\begin{equation}\label{eq: D-gen}
D_{\text{gen}} = \gamma^{0} \frac{\partial}{\partial t} - \gamma^{x} \frac{\partial}{\partial x} - \gamma^{z} \sqrt{q} \, \frac{\partial}{\partial z}.
\end{equation}
Note that the $y$-direction is absent due to $g^{22}=0$, reflecting the fact that the $q$-generalized Heisenberg algebra involves only two noncommutative generators.
\end{example}

\begin{example}[$q$-$\hbar$ Heisenberg algebra \cite{Volovich-Arefeva-91}]\label{ex:qhbar-dirac}
The $q$-$\hbar$ Heisenberg algebra, introduced in Example~\ref{ex:q-h}, is characterized by the metric components listed in Table~\ref{tab:metric-qhbar-dirac}. For each choice of the spatial indices $j,k \in \{1,2,3\}$, the corresponding Dirac operator $D_{q-\hbar}$ is obtained by substituting the metric components into the general expression \eqref{eq: Dirac}.

\begin{table}[H]
\centering
\caption{Metric components and associated Dirac operators for the $q$-$\hbar$-deformed Heisenberg algebra}
\label{tab:metric-qhbar-dirac}
\begin{tabular}{|c|c|c|c|c|c|c|}
\hline
$j$ & $k$ & $g^{00}$ & $g^{11}$ & $g^{22}$ & $g^{33}$ & $D_{q-\hbar}$ \\
\hline
1 & 1 & $-q$ & $1$ & $q^{1/2}$ & $0$ & $\gamma^{0} \sqrt{q} \, \partial_t - \gamma^{x} \partial_x - \gamma^{y} q^{1/4} \partial_y$ \\
\hline
2 & 2 & $-q$ & $1$ & $0$ & $q^{1/2}$ & $\gamma^{0} \sqrt{q} \, \partial_t - \gamma^{x} \partial_x - \gamma^{z} q^{1/4} \partial_z$ \\
\hline
3 & 3 & $-q$ & $q^{1/2}$ & $0$ & $1$ & $\gamma^{0} \sqrt{q} \, \partial_t - q^{1/4} \gamma^{x} \partial_x - \gamma^{z} \partial_z$ \\
\hline
\end{tabular}
\end{table}

In the table, $\partial_t = \frac{\partial}{\partial t}$, $\partial_x = \frac{\partial}{\partial x}$, $\partial_y = \frac{\partial}{\partial y}$, $\partial_z = \frac{\partial}{\partial z}$, and $\gamma^0, \gamma^x, \gamma^y, \gamma^z$ are the Dirac gamma matrices satisfying the Clifford algebra \eqref{eq: clifford}. The absence of a spatial direction in some cases (e.g., $g^{22}=0$ or $g^{33}=0$) reflects the fact that the deformation affects only a subset of the spatial coordinates.

One verifies directly that each Dirac operator satisfies $D_{q-\hbar}^2 = \square_q \mathbf{1}_4$, with
\[
\square_q = q \frac{\partial^2}{\partial t^2} - \frac{\partial^2}{\partial x^2} - q^{1/2} \frac{\partial^2}{\partial y^2} \quad \text{or} \quad \square_q = q \frac{\partial^2}{\partial t^2} - \frac{\partial^2}{\partial x^2} - q^{1/2} \frac{\partial^2}{\partial z^2},
\]
depending on the case.
\end{example}

\begin{remark}[Unification]
The three examples above illustrate how the general $q$-Dirac operator $D$ defined in \eqref{eq: Dirac} specializes to different $q$-deformed Dirac operators depending on the underlying $q$-Heisenberg algebra. The metric components $g^{\mu\nu}$ serve as the unifying parameters that encode the deformation. In all cases, the factorization property $D^2 = \square_q \mathbf{1}_4$ holds, establishing a direct link between the $q$-deformed Dirac equation and the $q$-deformed Klein–Gordon equation.
\end{remark}
\begin{example}[Deformed Heisenberg algebra from metric coefficients \cite{Jar-Nec-Hal-2025}]\label{ex:deformed-spacetime}
Sylvester's theorem of inertia guarantees that any Lorentzian metric can be diagonalized. Starting from the general quadratic form
\[
g = |\alpha| \, dx^0 \otimes dx^0 - |\beta| \, dx \otimes dx - |\gamma| \, dy \otimes dy - |\delta| \, dz \otimes dz,
\]
we introduce the rescaled coordinates
\[
\chi = \sqrt{|\alpha|}\,x^0,\qquad 
\xi = \sqrt{|\beta|}\,x,\qquad 
\eta = \sqrt{|\gamma|}\,y,\qquad 
\zeta = \sqrt{|\delta|}\,z.
\]
The metric then becomes the standard Minkowski metric
\[
g = d\chi \otimes d\chi - d\xi \otimes d\xi - d\eta \otimes d\eta - d\zeta \otimes d\zeta,
\]
with signature $(+1,-1,-1,-1)$, which is invariant under such transformations~\cite{greub1978linear, geroch1985mathematical}.

\medskip
\noindent \textbf{Deformation of the Heisenberg algebra.}
The spatial coefficients $\beta,\gamma,\delta$ serve as deformation parameters. Using the rescaled spatial coordinates
\[
\xi = \sqrt{|\beta|}\,x,\qquad \eta = \sqrt{|\gamma|}\,y,\qquad \zeta = \sqrt{|\delta|}\,z,
\]
and their conjugate momenta $\hat{p}_x,\hat{p}_y,\hat{p}_z$, we obtain the deformed commutation relations
\begin{align}
\xi \,\hat{p}_x - \hat{p}_x \,\xi &= i\sqrt{|\beta|},\\
\eta \,\hat{p}_y - \hat{p}_y \,\eta &= i\sqrt{|\gamma|},\\
\zeta \,\hat{p}_z - \hat{p}_z \,\zeta &= i\sqrt{|\delta|}.
\end{align}
When $|\beta|=|\gamma|=|\delta|=1$, the standard Heisenberg algebra
\[
[\xi,\hat{p}_x]=i,\qquad [\eta,\hat{p}_y]=i,\qquad [\zeta,\hat{p}_z]=i
\]
is recovered.

\medskip
\noindent \textbf{Associated Dirac operator.}
The Dirac operator in the deformed coordinates is given by
\[
D_{st} = \gamma_{\chi} \frac{\partial}{\partial \chi} - \gamma_{\xi} \frac{\partial}{\partial \xi} - \gamma_{\eta} \frac{\partial}{\partial \eta} - \gamma_{\zeta} \frac{\partial}{\partial \zeta},
\]
where the deformed gamma matrices are defined as
\[
\gamma_{\chi} = \sqrt{|\alpha|}\,\gamma^0,\quad 
\gamma_{\xi} = \sqrt{|\beta|}\,\gamma^{x},\quad 
\gamma_{\eta} = \sqrt{|\gamma|}\,\gamma^{y},\quad 
\gamma_{\zeta} = \sqrt{|\delta|}\,\gamma^{z}.
\]
The metric components in these coordinates are $g^{00}=1$, $g^{11}=g^{22}=g^{33}=-1$, which correspond to the Lorentzian signature. This construction illustrates how the metric coefficients directly enter both the commutation relations and the Dirac operator, bridging spacetime geometry with deformed quantum algebras \cite{connes1994noncommutative, douglas2001noncommutative, seiberg1999string}.

\medskip
\noindent \textbf{Connection with $M_1$ and $M_2$.}
As will be seen below, this construction is a particular case of the more general metric-deformed Heisenberg algebras $M_1$ and $M_2$ introduced in Definition~\ref{eq:first} and \ref{eq:second}.
\end{example}

\section{Conclusions and Outlook}\label{sec:conclusions} 

In this work, we have introduced metric-deformed Heisenberg algebras $M_1$ and $M_2$, where the commutation relations are expressed directly in terms of the components of a diagonal Lorentzian metric. We have shown that these algebras unify several known $q$-deformed Heisenberg algebras, including the $q$-$\hbar$ Heisenberg algebra \cite{Volovich-Arefeva-91}, the new $q$-Heisenberg algebra \cite{jaramillo2025new}, and the $q$-generalized Heisenberg algebra \cite{Razavinia-Lopes2022}. This unification demonstrates that the metric components $g^{\mu\nu}$ serve as natural deformation parameters for the Heisenberg algebra, providing a geometric interpretation of $q$-deformations within the framework of special relativity. Using Sylvester's theorem of inertia, we established a direct connection between the metric signature and the deformation parameters. The diagonalization procedure guarantees that any Lorentzian metric can be reduced to a form with entries $\pm 1$, and the coefficients that appear during this process become the deformation parameters in the commutation relations of $M_1$ and $M_2$. This construction reveals a deep link between the algebraic structure of quantum mechanics and the geometry of spacetime. Furthermore, we constructed a $q$-Dirac operator $D_q$ from the deformed D'Alembertian and proved the factorization property $D_q^2 = \square_q \mathbf{1}_4$, establishing a direct link between the $q$-deformed Dirac equation and the $q$-deformed Klein–Gordon equation. This construction is compatible with our previous work on quadratic $q$-difference operators \cite{Jaramillo2024}, where we developed the analytic tools for $q$-derivatives on quadratic relativistic invariants. The three explicit realizations presented in Section~\ref{sec:examples} illustrate how the general $q$-Dirac operator specializes to different cases depending on the underlying $q$-Heisenberg algebra.

Several avenues for further research emerge from this work. First, the representation theory of $M_1$ and $M_2$ on $q$-deformed Hilbert spaces remains to be explored, generalizing the work of \cite{schmudgen-09, Silvestrov-Hellstrom-00}. Second, the $q$-deformed Dirac equation $D_q \psi = m \psi$ can be solved using $q$-exponential functions, which would yield a $q$-deformed dispersion relation with potential phenomenological implications. Third, the $q$-Dirac operator can be extended to Clifford algebra-valued functions, following the approach of \cite{Jaramillo2024}, which would allow for a full $q$-deformed Clifford analysis. Fourth, a $q$-deformed version of the relativistic Maxwell electrodynamic algebra proposed in \cite{Jaramillo2024} could be developed. Fifth, the physical implications of identifying the deformation parameter $q$ with metric components may lead to testable predictions for quantum gravity effects. Finally, it would be interesting to investigate whether the triple $(\mathcal{A}, \mathcal{H}, D_q)$ forms a $q$-deformed spectral triple in the sense of Connes' noncommutative geometry.

\section*{Statements and Declarations}

\noindent \textbf{Competing Interests:} The authors declare no competing interests.

\noindent \textbf{Data Availability:} No new data were created or analyzed in this study. Data sharing is not applicable to this article.

\end{document}